\documentclass{article}
\usepackage{authblk}
\usepackage{amsmath,amssymb,amsthm,amsfonts,bm,fancybox,fullpage}
\usepackage{multirow}
\usepackage{color}
\usepackage[pdftex]{graphicx}
\usepackage{cite}

\theoremstyle{plain}
\newtheorem{theorem}{Theorem}[section]
\newtheorem{lemma}[theorem]{Lemma}
\newtheorem{corollary}[theorem]{Corollary}

\theoremstyle{definition}
\newtheorem{definition}[theorem]{Definition}

\newtheorem{remark}[theorem]{Remark}

\newcommand{\ra}{\rightarrow}
\newcommand{\la}{\leftarrow}

\newcommand{\bin}{\{0,1\}}
\newcommand{\N}{\mathbb{N}}
\newcommand{\Z}{\mathbb{Z}}

\newcommand{\QR}{\mathcal{R}}
\newcommand{\QN}{\mathcal{N}}
\newcommand{\seq}{\mathsf{S}}
\newcommand{\x}{x}

\newcommand{\Enc}{\mathsf{Enc}}
\newcommand{\Dec}{\mathsf{Dec}}

\newcommand{\qr}[1]{(\frac{#1}{p})}

\newcommand{\Qr}[2]{\left(\dfrac{#2}{#1}\right)}
\newcommand{\dare}[1]{{\sf D}(#1)}

\begin{document}

\title{Private Simultaneous Messages Based on Quadratic Residues}
\author[1,5]{Kazumasa Shinagawa}
\author[2,5]{Reo Eriguchi}
\author[3]{Shohei Satake}
\author[4,5]{Koji Nuida}
\affil[1]{Ibaraki University}
\affil[2]{The University of Tokyo}
\affil[3]{Meiji University}
\affil[4]{Institute of Mathematics for Industry (IMI), Kyushu University}
\affil[5]{National Institute of Advanced Industrial Science and Technology (AIST)}
%
\date{}
\maketitle

\begin{abstract}
Private Simultaneous Messages (PSM) model is a minimal model for secure multiparty computation. 
Feige, Kilian, and Naor (STOC 1994) and Ishai (Cryptology and Information Security Series 2013) constructed PSM protocols based on quadratic residues. 
In this paper, we define QR-PSM protocols as a generalization of these protocols. 
A QR-PSM protocol is a PSM protocol whose decoding function outputs the quadratic residuosity of what is computed from messages. 
We design a QR-PSM protocol for any symmetric function $f: \bin^n \ra \bin$ of communication complexity $O(n^2)$. 
As far as we know, it is the most efficient PSM protocol since the previously known best PSM protocol was of $O(n^2\log n)$ (Beimel et al., CRYPTO 2014). 
We also study the sizes of the underlying finite fields $\mathbb{F}_p$ in the protocols since the communication complexity of a QR-PSM protocol is proportional to the bit length of the prime $p$. 
In particular, we show that the $N$-th Peralta prime $P_N$, which is used for general QR-PSM protocols, can be taken as at most $(1+o(1))N^22^{2N-2}$, which improves the Peralta's known result (Mathematics of Computation 1992) by a constant factor $(1+\sqrt{2})^2$. 
\end{abstract}

\section{Introduction}

\emph{Private Simultaneous Messages (PSM) model} introduced by Feige, Kilian, and Naor \cite{FKN94} and named by Ishai and Kushilevitz \cite{IK97} is a minimal model for non-interactive secure multiparty computation with information-theoretic security. 
In the PSM model, there are $n$ \emph{players} and a special party called a \emph{referee}. 
Each player $P_i$ computes a message $m_i$ from $P_i$'s input $x_i$ and a \emph{shared randomness} $r$, and sends $m_i$ to the referee. 
Here, the shared randomness $r$ is known by all players but the referee. 
Given $n$ messages $m_1, m_2, \ldots, m_n$, the referee computes an output value $y$, which is expected to be $y = f(x_1, x_2, \ldots, x_n)$ for a function $f$ agreed upon by all players and the referee. 
The security of the protocol ensures that the referee cannot learn anything about the secret inputs beyond what can be inferred from the output value. 
The efficiency of PSM protocols is mainly measured by the communication complexity $\sum_{i=1}^n |m_i|$, where $|\cdot|$ denotes the bit length. 

\subsection{PSM Protocols Based on Quadratic Residues}

In this paper, we study PSM protocols based on quadratic residues. 
As far as we know, there are two existing PSM protocols based on quadratic residues. 
Feige, Kilian, and Naor \cite{FKN94} proposed such a protocol for comparing two numbers $x$ and $y$, i.e., deciding whether $x \geq y$ or not. 
Ishai \cite{Ishai13} designed such a protocol for any function $f: \bin^n \ra \bin$. 

\subsubsection{Feige-Kilian-Naor's Protocol}\label{ss:FKN}

The protocol based on quadratic residues by Feige, Kilian, and Naor \cite{FKN94} is a two-player PSM protocol computing the comparison function ${\sf COMP}: \{0,1,2\} \times \{0,1,2\} \ra \{-1, 0, 1\}$ as follows:
\[
{\sf COMP}(x_1, x_2) = 
\begin{cases}
1 & \text{if $x_1 > x_2$},\\
0 & \text{if $x_1 = x_2$},\\
-1 & \text{if $x_1 < x_2$}.
\end{cases}
\]
The shared randomness of the protocol is a pair $(r_1,r_2)$ of an element $r_1$ of $\Z/7\Z$ and a nonzero quadratic residue $r_2$ modulo $7$. 
The first player $P_1$ computes a message $m_1 \in \Z/7\Z$ as $m_1 := r_1 + r_2 x_1 \; (\bmod\;7)$, and the second player $P_2$ computes a message $m_2 \in \Z/7\Z$ as $m_2 := -r_1 - r_2 x_2 \: (\bmod\;7)$. 
Given $m_1, m_2$, the referee computes the quadratic residuosity of $m := m_1 + m_2\; (\bmod\;7)$, and outputs $1$ if $m$ is a non-zero quadratic residue, $-1$ if $m$ is a quadratic nonresidue, and $0$ if $m = 0$. 

\subsubsection{Ishai's Protocol}\label{ss:Ishai}

The protocol based on quadratic residues by Ishai \cite{Ishai13} is an $n$-player PSM protocol computing any function $f: \bin^n \ra \bin$. 
Let $p$ be a prime and $0 < a \leq p-2^n$ an integer such that $a+\sum_{i=1}^n 2^{i-1}b_i$ is a quadratic residue modulo $p$ if and only if $f(b_1, b_2, \ldots, b_n) = 1$. 
From the result by Peralta \cite{Per92}, such a prime $p$ with $p = 2^{O(n)}$ exists. 
The shared randomness of the protocol is a tuple $(r_0, r_1, r_2, \ldots, r_n)$ of a nonzero quadratic residue $r_0$ modulo $p$ and $r_1, r_2, \ldots, r_n \in \Z/p\Z$ such that $\sum_{i=1}^n r_i \equiv 0\; (\bmod\;p)$. 
The player $P_i$ holding $x_i \in \bin$ computes a message $m_i \in \Z/p\Z$ as $m_i := 2^{i-1} r_0 x_i + r_i \; (\bmod\;p)$ if $2 \leq i \leq n$ and $m_1 := r_0 (a + x_1) + r_1 \; (\bmod\;p)$ if $i=1$. 
Given $m_1, m_2, \ldots, m_n$, the referee computes the quadratic residuosity of $m := \sum_{i=1}^n m_i \; (\bmod\;p)$, and outputs $1$ if $m$ is a quadratic residue and $0$ otherwise. 

\subsection{Our Contributions}

First, we introduce the notions of \emph{quadratic residue based PSM (QR-PSM) protocols} and \emph{linear QR-PSM (LQR-PSM) protocols}. 
Let $p$ be a prime. 
A QR-PSM protocol modulo $p$ is a PSM protocol such that the decoding function of the protocol outputs the quadratic residuosity (Legendre symbol) of $\phi(m_1, m_2, \ldots, m_n)$ modulo $p$, where $\phi$ is a function from messages to $\Z/p\Z$, and $m_i$ is the $i$-th message for $1 \leq i \leq n$. 
An LQR-PSM protocol modulo $p$ is a QR-PSM protocol modulo $p$ such that $\phi(m_1, m_2, \ldots, m_n) = \sum_{i=1}^n m_i \;(\bmod\;p)$. 
We remark that Feige-Kilian-Naor's protocol and Ishai's protocol are LQR-PSM protocols. 

Next, we construct new QR-PSM and LQR-PSM protocols. 
For any symmetric function $f: \bin^n \ra \bin$, we obtain an LQR-PSM protocol of communication complexity $O(n^2)$. 
We note that it is the most efficient PSM protocol for symmetric functions so far since the previously known best protocol was of $O(n^2\log n)$ proposed by Beimel, Gabizon, Ishai, Kushilevitz, Meldgaard, and Paskin-Cherniavsky \cite{BGIKMP14}. 
For any weighted threshold function $f: \bin^n \ra \bin$ with weight vector $\bm{w}$ and threshold $t$, we also obtain an LQR-PSM protocol of communication complexity $O(n \cdot \sum_{i=1}^n |w_i|)$. 
We remark that these protocols are more efficient than the protocols obtained by applying Ishai's protocol to these specific functions (see Table \ref{tab:comparison} for efficiency comparison). 
In addition, we show that QR-PSM protocols can be obtained from decomposable randomized encodings (DRE). 
In particular, we show that if a function $f$ is ``embedded'' into another function $g$ and $g$ admits a DRE of output length $s$, we have a QR-PSM protocol with communication complexity $O(s\cdot l(g))$, where $l(g)$ is the ``embedding length'' of $g$ (see Section \ref{ss:linear} for the definition of the embedding). 
This construction can be viewed as a generalization of our LQR-PSM protocols since it admits not only linear polynomials but also higher-degree polynomials. 

In QR-PSM protocols, the communication complexity is dominated by the size of modulus $p$. 
Thus, it is important to give upper and lower bounds on the primes. 
We study two kinds of primes which we name the \emph{Peralta primes} and the \emph{LQR-PSM primes}: the $n$-th Peralta prime $P_n$ is the smallest prime $p$ such that every $n$-bit string appears in the ``quadratic residue sequence modulo $p$'' as a subsequence; and the $n$-th LQR-PSM prime $L_n$ is the smallest prime $p$ such that every function $f: \bin^n \ra \bin$ has an LQR-PSM protocol modulo $p$. 
We first show that $L_n \leq P_{2^{n-1}}$ and $L_n \geq 2^{\frac{2^n-2}{n}}$. 
Since it holds $P_n \leq (1+\sqrt{2})^2n^22^{2n-2}$ for sufficiently large $n$ from the result by Peralta~\cite{Per92}, we have upper and lower bounds on the LQR-PSM primes. 
We also show a slightly better upper bound on the Peralta primes by using graph theory. 
In particular, we have $P_n \leq (1+o(1))n^22^{2n-2}$, which improves the Peralta's result by a constant factor $(1+\sqrt{2})^2$. 

\begin{table}[t]
\caption{The communication complexity of QR-PSM protocols (see Section \ref{ss:linear} for the notations)}
\begin{center}
\begin{tabular}{c|c|c} \hline
 & function & communication complexity \\ \hline
Ishai~\cite{Ishai13} & any function & $O(n \cdot 2^n)$ \\
Corollary \ref{corollary:symmetric} & symmetric function & $O(n^2)$  \\
Corollary \ref{corollary:weighted} & weighted threshold function with weight vector $\bm{w}$ & $O(n \cdot \sum_{i=1}^n |w_i|)$  \\ \hline
\end{tabular}
\end{center}
\label{tab:comparison}
\end{table}%

\subsection{Related Work}

The PSM model was firstly introduced by Feige, Kilian, and Naor \cite{FKN94}. 
Besides the QR-PSM protocol described in Section \ref{ss:FKN}, they also constructed a two-player PSM protocol for any function $f: \bin^n \times \bin^n \ra \bin$ of complexity $O(2^n)$. 
Beimel, Ishai, Kumaresan, and Kushilevitz \cite{BIKK14} improved it to $O(2^{n/2})$ by introducing a decomposable private information retrieval protocol. 
This is still the state-of-the-art two-player PSM protocol among those applicable to arbitrary function $f: \bin^n \times \bin^n \ra \bin$. 
As an impossibility result, Applebaum, Holenstein, Mishra, and Shayevitz \cite{AHMS20} showed that any two-player PSM protocol computing a random function $f: \bin^n \times \bin^n \ra \bin$ requires the complexity $3n-O(\log n)$. 
Narrowing this exponential gap between the upper and lower bounds is an important open problem in cryptography \cite{Vaikuntanathan18}. 

For the case of $k$ players for $k \geq 3$, Beimel, Kushilevitz, and Nissim \cite{BKN18} constructed a $k$-player PSM protocol for any function $f: (\bin^n)^k \ra \bin$ of complexity $O(poly(k) \cdot 2^{nk/2})$. 
Assouline and Liu \cite{AL21} improved it to $O(2^{n(k-1)/2})$ for infinitely many $k$'s and conjectured that it holds for any $k$. 

For a specific class of functions, Ishai and Kushilevitz \cite{IK97} constructed a PSM protocol for a Boolean modulo-$p$ branching program $BP: \bin^n \ra \bin$ of size $a$ with communication complexity $O(\log p \cdot n \cdot a^2)$. 

\section{Preliminaries}\label{s:preliminaries}

\subsection{Notations}
For an integer $n \geq 2$, we denote $[n] := \{1, 2, \ldots, n\}$ and $\Z_n := \Z/n\Z$. 
For a set $S$, we denote by $\# S$ the cardinality of $S$.  
For a bit string $m \in \bin^*$, we denote by $|m|$ the bit length of $m$. 
For an integer $a \in \Z$, we denote by $|a|$ the absolute value of $a$. 

Let $A$ be a ring. 
An \emph{arithmetic formula} over $A$ is a rooted binary tree, where each leaf is labeled by either an input variable $x_i$ ($1 \leq i \leq n$) or a constant $c \in A$, and each intermediate node called a \emph{gate} is labeled by either addition or multiplication. 
Its \emph{depth} is defined by the length of the longest path from the root to a leaf. 
An arithmetic formula can be regarded as a function $f: A^n \ra A$ naturally. 
A \emph{Boolean formula} is an arithmetic formula over $A = \Z_2$. 
In this paper, the basis of Boolean formulas is always $\{\wedge, \oplus\}$. 

A \emph{polynomial} over $A$ is a polynomial whose coefficients are elements of $A$. 
A polynomial can be regarded as a function $f: A^n \ra A$ naturally. 
Every arithmetic formula over $A$ can be represented by a polynomial over $A$. 
In particular, every Boolean formula $f: \bin^n \ra \bin$ can be represented by a polynomial over $\Z_2$ (over the basis $\{\wedge, \oplus\}$), which is called the \emph{Reed-Muller canonical form} of $f$. 

\subsection{PSM Protocols}

\begin{definition}[PSM protocol]
Let $n \geq 2$ be an integer, and $X_1, X_2, \ldots, X_n, Y, R, M_1, M_2, \ldots, M_n$ finite sets. 
Set $X = \prod_{1\leq i \leq n}X_i$ and $M = \prod_{1\leq i \leq n}M_i$. 
Let $\Enc_i: X_i \times R \ra M_i$ ($1 \leq i \leq n$) and $\Dec: M \ra Y$ be functions. 
Here, $X_i, Y, R, M_i, \Enc_i, \Dec$ ($1 \leq i \leq n$) are called the \emph{$i$-th input space}, the \emph{output space}, the \emph{randomness space}, the \emph{$i$-th message space}, the \emph{$i$-th encoding function}, and the \emph{decoding function}, respectively. 
A \emph{private simultaneous messages (PSM) protocol $\Pi$} for a function $f: X \ra Y$ is a $7$-tuple 
\[
\Pi = (n, X, Y, R, M, (\Enc_i)_{1\leq i \leq n}, \Dec),
\]
satisfying the following conditions:
\begin{description}
\item[Correctness.] 
For any $(x_1, \ldots, x_n) \in X$ and any $r \in R$, it holds that
\[
\Dec((\Enc_1(x_1, r),\ldots, \Enc_n(x_n, r))= f(x_1,  \ldots, x_n).
\]
\item[Security.] 
For any $m \in M$ and $x = (x_1,  \ldots, x_n), x' = (x'_1, \ldots, x'_n) \in X$ with $f(x) = f(x')$, it holds that 
\[
\Pr_{r \in R}\bigl[ (\Enc_1(x_1, r),\ldots, \Enc_n(x_n, r)) = m \bigr] =\Pr_{r \in R}\bigl[ (\Enc_1(x'_1, r), \ldots, \Enc_n(x'_n, r)) = m \bigr],
\]
where $r \in R$ is chosen uniformly at random. 
\end{description}
The communication complexity is defined by $\sum_{i=1}^n \log_2(\# M_i)$ and the randomness complexity is defined by $\log_2(\# R)$. 
\end{definition}

\subsection{Decomposable Randomized Encodings}

In this section, we define the notions of \emph{randomized encodings} and \emph{decomposable randomized encodings (DRE)}. 
A DRE over $\Z_p$ for a prime $p$ is used as a building block for constructing QR-PSM protocols. 

\begin{definition}[Randomized encoding]
Let $X, Y, \hat{Y}, R$ be finite sets, and $f: X \ra Y$ a function. 
A \emph{randomized encoding} $\hat{f}: X \times R \ra \hat{Y}$ is a function satisfying the following conditions:
\begin{description}
\item[Correctness.] 
There exists a function $\Dec: \hat{Y} \ra Y$ called a \emph{decoder} such that for any $x \in X$ and $r \in R$, it holds $\Dec(\hat{f}(x, r)) = f(x)$. 
\item[Security.] For any $\hat{y} \in \hat{Y}$ and $x, x' \in X$ such that $f(x) = f(x')$, it holds that 
\[
\Pr_{r \in R}\bigl[ \hat{f}(x, r) = \hat{y} \bigr] =\Pr_{r \in R}\bigl[ \hat{f}(x', r) = \hat{y} \bigr],
\]
where $r \in R$ is chosen uniformly at random. 
\end{description}
\end{definition}

\begin{definition}[DRE]
Let $A$ be a finite ring, and $f: A^n \ra A$ a function. 
A \emph{decomposable randomized encoding (DRE)} of $f$ is a randomized encoding $\hat{f}: A^n \times A^m \ra A^s$ as follows:
\[
\hat{f}((x_1, x_2, \ldots, x_n), r) = (\hat{f}_0(r), \hat{f}_1(x_1, r), \hat{f}_2(x_2, r), \ldots, \hat{f}_{n}(x_{n}, r))
\]
where $\hat{f}_0: A^m \ra A^{s_0}$ and $\hat{f}_i:  A \times A^m \ra A^{s_i}$ ($1\leq i \leq n$) are functions such that $\sum_{i=0}^{n} s_i = s$. 
The integer $s$ is called the \emph{output length} of the DRE. 
\end{definition}

For a function $f: \Z^n \ra \Z$, we define the \emph{DRE complexity} of $f$. 

\begin{definition}[DRE complexity]
Let $f: \Z^n \ra \Z$ be a function. 
For a prime $p$, define $f_p: \Z_{p}^n \ra \Z_{p}$ as the function such that $f_p \equiv f\;(\bmod\;p)$. 
The \emph{DRE complexity} of $f$, denoted by $\dare{f}$, is defined by the minimum integer $s$ such that for every prime $p$, there exists a DRE of $f_p$ with output length at most $s$. 
\end{definition}

Based on Cleve's result \cite{Cleve90} on straight-line programs, Cramer, Fehr, Ishai, and Kushilevitz designed a constant-round multiparty computation protocol for arithmetic formulas \cite[Theorem 3]{CFIK03}. 
This construction can be viewed as a DRE of arithmetic formulas. 

\begin{theorem}[{Cramer-Fehr-Ishai-Kushilevitz \cite{CFIK03}}]\label{thm:DRE}
Let $f: A^n \ra A$ be an arithmetic formula of depth $d$. 
Then, there exists a DRE of $f$ with output length $2^{d+O(\sqrt{d})}$. 
\end{theorem}

\begin{corollary}\label{corollary:DRE}
Let $f: \Z^n \ra \Z$ be an arithmetic formula of depth $d$. 
Then, we have $\dare{f} \leq 2^{d+O(\sqrt{d})}$. 
\end{corollary}

Based on Theorem \ref{thm:DRE}, we have a DRE of polynomials. 

\begin{theorem}\label{thm:DRE_poly}
Let $f: A^n \ra A$ be a degree-$k$ polynomial having $m$ terms. 
Then, there exists a DRE of $f$ with output length $m\cdot k \cdot 2^{O(\sqrt{\log k})}$. 
\end{theorem}

\begin{proof}
Let $g: A^{k+1} \ra A$ be a function such that $g(y_0, y_1, \ldots, y_k) = y_0 + \prod_{i=1}^{k}y_i$. 
Since $g$ can be represented by an arithmetic formula of depth $d = \lceil \log_2 k \rceil + 1$, it has a DRE with output length $2^{d + O(\sqrt{d})} = k \cdot 2^{O(\sqrt{\log k})}$ from Theorem \ref{thm:DRE}. 
Suppose that the $i$-th term of $f$ is a degree-$k'$ term of the form $cx_{j_1}x_{j_2}\cdots x_{j_{k'}}$ ($c \in A$, $k' \leq k$). 
Let $r_1, r_2, \ldots, r_m \in A$ be random numbers such that $\sum_{i=1}^m r_i = 0$. 
Then, we have a DRE of $cx_{j_1}x_{j_2}\cdots x_{j_{k'}} + r_i$ from the DRE of $g$, by setting 
\[
(y_0, y_1, y_2, \ldots, y_{k'}, y_{k'+1}, y_{k'+2}, \ldots, y_k) \la (r_i, cx_{j_1}, x_{j_2}, \ldots, x_{j_{k'}}, 1, 1, \ldots, 1).
\] 
Juxtaposing them for each term, we obtain the DRE of $f$ with output length $m \cdot k \cdot 2^{O(\sqrt{\log k})}$. 
\end{proof}

Let $f: A^n \ra A$ be a degree-$k$ polynomial having $m$ terms. 
Since $f$ can be represented by an arithmetic formula of depth $d = \lceil \log_2 k \rceil + \lceil \log_2 m \rceil$, Theorem \ref{thm:DRE} results in a DRE of $f$ with output length $2^{\log_2 d + O(\sqrt{d})} = m \cdot k \cdot 2^{O(\sqrt{\log k + \log m})}$. 
On the other hand, Theorem \ref{thm:DRE_poly} results in a DRE of $f$ with output length $m \cdot k \cdot 2^{O(\sqrt{\log k})}$. 
Thus, Theorem \ref{thm:DRE_poly} is more efficient than Theorem \ref{thm:DRE} by the factor $2^{O(\sqrt{\log m})}$ in this case. 

\subsection{Quadratic Residues}\label{ss:sequence}

We denote by $\QR_p \subset \Z_p$ the set of non-zero quadratic residues modulo $p$ and by $\QN_p \subset \Z_p$ the set of quadratic nonresidues modulo $p$. 
For an integer $a \in \Z$, the \emph{Legendre symbol} $\qr{a}$ is defined as follows:
\[
\Qr{p}{a} = 
\begin{cases}
1 & \text{if $a \not\equiv 0 \;(\bmod \;p)$ is a quadratic residue modulo $p$,}\\
0 & \text{if $a \equiv 0 \;(\bmod \;p)$,}\\
-1 & \text{if $a$ is a quadratic nonresidue modulo $p$.}\\
\end{cases}
\]

For a prime $p$, we define \emph{the quadratic residue sequence modulo $p$} as the string $\seq_p \in \bin^{p-1}$ such that for every $i \in [p-1]$, the $i$-th bit (from the left) of $\seq_p$ is equal to $1$ if $\qr{i} = 1$ and $0$ otherwise. 
If a string $t\in \bin^*$ is a substring of $\seq_p$, then we say that \emph{$\seq_p$ contains $t$}. 
The quadratic residue sequences modulo primes from $2$ to $19$ are shown as follows:
\[
\begin{tabular}{c|l}\hline
$p$ & $\seq_p$ \\ \hline
$2$&$1$\\
$3$&$10$\\
$5$&$1001$\\
$7$&$110100$\\
$11$&$1011100010$\\
$13$&$101100001101$\\
$17$&$1101000110001011$\\
$19$&$100111101010000110$\\ \hline
\end{tabular}
\]

By Weil's character sum estimation over finite fields, Peralta~\cite{Per92} gave a sufficient condition on primes for containing every $n$-bit string $t \in \bin^n$. 

\begin{theorem}[Peralta~\cite{Per92}]\label{thm:peralta}
Let $p$ be a prime. 
If $p\cdot \left(\frac{1}{2}\right)^n> n(\sqrt{p} + 3)$, then $\seq_p$ contains every $n$-bit string $t \in \bin^n$. 
\end{theorem}

We say that a prime $p$ is \emph{$n$-Peralta} if $\seq_p$ contains every $n$-bit string $t \in \bin^n$. 
We define \emph{the $n$-th Peralta prime $P_n$} as the smallest $n$-Peralta prime. 
The $n$-th Peralta primes for $1 \leq n \leq 8$ obtained by computer experiments are shown as follows:
\[
\begin{tabular}{c|cccccccc}\hline
$n$ & 1 & 2 & 3 & 4 & 5 & 6 & 7 & 8 \\ \hline
$P_n$ & 3 & 7 & 11 & 37 & 67 & 181 & 367 & 1091 \\ \hline
\end{tabular}
\]

Applying the Baker-Harman-Pintz theorem on prime gaps in \cite{BHP2001}, we obtain the following corollary. 

\begin{corollary}
\label{cor:peralta}
For any sufficiently large $n$, there exists an $n$-Peralta prime $p$ with $p \leq c+c^{0.525}$, where $c = (1+\sqrt{2})^2 n^2 2^{2n-2}$. 
Hence, $\log P_n = O(n)$ holds. 
\end{corollary}

\begin{proof}
From Theorem \ref{thm:peralta}, any prime $p$ satisfying
\begin{equation*}\label{eq:Peralta1}
\sqrt{p} > n2^{n-1} + \sqrt{n^22^{2n-2} + 3n2^n}
\end{equation*}
is $n$-Peralta. 
As $\sqrt{2}n2^{n-1} > \sqrt{n^22^{2n-2} + 3n2^n}$ for all $n \geq 3$, any prime $p$ satisfying
\[
\sqrt{p} > (1+\sqrt{2})n2^{n-1}
\]
is also $n$-Peralta. 
By the Baker-Harman-Pintz theorem on prime gaps, there exists a prime $p$ in $[c, c+c^{0.525}]$ for $c = (1+\sqrt{2})^2 n^2 2^{2n-2}$, as desired. 
\end{proof}

In Section \ref{ss:upper}, we improve the upper bound on Peralta primes by a constant factor $(1+\sqrt{2})^2$. 

\section{QR-PSM Protocols}\label{s:QRP}

\subsection{Definition of QR-PSM Protocols}\label{ss:Def_QRPSM}

We define quadratic residue based PSM protocols. 
It is a PSM protocol whose decoding function outputs the Legendre symbol of what is computed from messages. 

\begin{definition}[QR-PSM protocol]\label{definition:QR-PSM}
Let $\Pi = (n, X, Y, R, M, (\Enc_i)_{1\leq i \leq n}, \Dec)$ be a PSM protocol such that $Y = \{-1, 0, 1\}$. 
Let $p$ be a prime. 
We say that $\Pi$ is a \emph{quadratic residue based PSM (QR-PSM) protocol modulo $p$} if there exists a function $\phi: M \ra \Z_p$ such that for any $(m_1, \ldots, m_n) \in M$, 
\[
\Dec(m_1, m_2, \ldots, m_n) = \Qr{p}{\phi(m_1, m_2, \ldots, m_n)}. 
\]
\end{definition}

We remark that Feige-Kilian-Naor's protocol (see Section \ref{ss:FKN}) is a QR-PSM protocol modulo $7$. 
We also point out that Ishai's protocol (see Section \ref{ss:Ishai}) is a QR-PSM protocol modulo a prime $p$. 

Let $f: \bin^n \ra \bin$ be a Boolean function. 
We say that a QR-PSM protocol computes $f$ if it outputs $(-1)^{f(x)}$. 
Throughout this paper, we focus on the QR-PSM protocols for Boolean functions in this sense. 

\subsection{LQR-PSM Protocols}\label{ss:linear}

We say that a function $f: \bin^n \ra \bin$ is \emph{embedded} into a function $g: \Z^n \ra \Z$ if $g(x) = g(x')$ implies $f(x) = f(x')$ for any $x, x' \in \bin^n$. 
The function $g$ is called an \emph{embedding of $f$}. 
The \emph{embedding length} of $g$, denoted by $l(g)$, is defined as follows: 
\[
l(g) := \max_{x \in \bin^n}(g(x)) - \min_{x \in \bin^n}(g(x)) + 1.
\] 

If a function $f: \bin^n \ra \bin$ can be embedded into a linear function $g = a_1 x_1 + a_2 x_2 + \cdots + a_n x_n$, we obtain an efficient QR-PSM protocol which we call a \emph{linear QR-PSM (LQR-PSM) protocol}. 

\begin{definition}[Linear QR-PSM protocol]
Let $p$ be a prime and $a_0, a_1, a_2, \ldots, a_n \in \Z_p$. 
\emph{A linear QR-PSM (LQR-PSM) protocol modulo $p$}, denoted by $[a_0, a_1, a_2, \ldots, a_n]_p$, is a QR-PSM protocol $\Pi = (n, \bin^n, \{-1, 0, 1\}, R, \Z_{p}^n, (\Enc_i)_{1\leq i \leq n}, \Dec)$ modulo $p$ in the following. 
\begin{itemize}
\item The randomness space $R$ is 
\[
R = \left\{(r_0, r_1, r_2, \ldots, r_n) \in \Z_p^{n+1} \; \middle| \; r_0 \in \QR_p,~\sum_{i = 1}^n r_i \equiv 0 \;(\bmod \;p) \right\}.
\]
\item The encoding function $\Enc_i: \bin \times R \ra \Z_p$ is 
\[
\Enc_i(x_i, r) = 
\begin{cases}
r_0(a_0 + a_i x_i) + r_i \;(\bmod \;p) & \text{if $i = 1$},~\\
r_0a_i x_i + r_i \;(\bmod \;p) & \text{otherwise},
\end{cases}
\]
where $r  = (r_0, r_1, r_2, \ldots, r_n) \in R$ and $x_i \in \bin$. 
\item The decoding function $\Dec: (\Z_p)^n \ra \{-1, 0, 1\}$ is
\[
\Dec(m_1, m_2, \ldots, m_n) = \Qr{p}{\sum_{i=1}^n m_i}.
\]
\end{itemize}
\end{definition}

\begin{theorem}\label{theorem:LQR-PSM}
Let $f: \bin^n \ra \bin$ be a function. 
Let $g: \Z^n \ra \Z$ be an embedding of $f$ such that $g = a_1 x_1 + a_2 x_2 + \cdots + a_n x_n$ for some $a_1, a_2, \ldots, a_n \in \Z$. 
Then, there exists an LQR-PSM protocol for $f$ with communication complexity $n \cdot \log_2 P_{l(g)}$, where $P_{l(g)}$ is the $l(g)$-th Peralta prime. 
\end{theorem}

\begin{proof}
Set $p := P_{l(g)}$.  
Since $p$ is the $l(g)$-th Peralta prime, there exists $a_0 \in \Z_p$ such that $\Qr{p}{a_0 + g(x)} = (-1)^{f(x)}$ for all $x \in \bin^n$. 
We claim that $[a_0, a_1, a_2, \ldots, a_n]_p$ is an LQR-PSM protocol for $f$. 
By setting $m_i := \Enc_i(x_i, r)$, we have 
\[
(m_1, m_2, \ldots, m_n) = (r_0(a_0 + a_1 x_1) + r_1, r_0a_2 x_2 + r_2, \ldots, r_0a_n x_n + r_n).
\]
Since $r_0$ is a nonzero quadratic residue, we have
\[
\Qr{p}{\sum_{i=1}^n m_i} = \Qr{p}{r_0(a_0 + a_1 x_1 + a_2 x_2 + \cdots + a_n x_n)} = \Qr{p}{r_0(a_0 + g(x))} =  \Qr{p}{a_0 + g(x)}.
\]
Thus, it correctly computes $f$. 
The communication complexity of the protocol is $n \cdot \log_2 P_{l(g)}$. 
\end{proof}

Theorem \ref{theorem:LQR-PSM} implies a protocol for any symmetric function. 

\begin{corollary}\label{corollary:symmetric}
For any symmetric function $f: \bin^n \ra \bin$, there exists an LQR-PSM protocol with communication complexity $n\cdot  \log_2 P_{n+1} = O(n^2)$. 
\end{corollary}

\begin{proof}
It follows from Theorem \ref{theorem:LQR-PSM} since any symmetric function is embedded to a linear function $g(x_1, x_2, \ldots, x_n) = x_1 + x_2 + \cdots + x_n$ of embedding length $n+1$. 
\end{proof}

\emph{A weighted threshold function $f_{\bm{w}, t}$ associated with $\bm{w} = (w_1, w_2, \ldots, w_n) \in \Z^n$ and $t \in \N$} is defined as 
\[
f_{\bm{w}, t}(x_1, x_2, \ldots, x_n) = 
\begin{cases}
1 & \text{if $\sum_{i=1}^n w_i x_i \geq t$,} \\
0 & \text{otherwise.}
\end{cases}
\]

\begin{corollary}\label{corollary:weighted}
For any $\bm{w}\in \Z^n$ and $t \in \N$, there exists an LQR-PSM protocol for the weighted threshold function $f_{\bm{w}, t}: \bin^n \ra \bin$ associated with $\bm{w}, t$ with communication complexity $n\cdot  \log_2 P_{W+1} = O(n\cdot W)$ for $W = \sum_{i=1}^n |w_i|$. 
\end{corollary}

\begin{proof}
It follows from Theorem \ref{theorem:LQR-PSM} since a weighted threshold function associated with $\bm{w}, t$ is embedded to a linear function $g(x_1, x_2, \ldots, x_n) = w_1x_1 + w_2x_2 + \cdots + w_nx_n$ of embedding length $\sum_{i=1}^n |w_i|+1$. 
\end{proof}

Theorem \ref{theorem:LQR-PSM} also implies Ishai's protocol (see Subsection \ref{ss:Ishai}). 

\begin{corollary}[Ishai \cite{Ishai13}]\label{corollary:any}
For any function $f: \bin^n \ra \bin$, there exists an LQR-PSM protocol with communication complexity $n\cdot  \log_2 P_{2^n} = O(n \cdot 2^n)$. 
\end{corollary}

\begin{proof}
It follows from Theorem \ref{theorem:LQR-PSM} since any function $f: \bin^n \ra \bin$ is embedded to a linear function $g(x_1, x_2, \ldots, x_n) = x_1 + 2x_2 + \cdots 2^{i-1}x_i+ \cdots + 2^{n-1}x_n$ of embedding length $2^n$. 
\end{proof}

We also obtain an LQR-PSM protocol for a composition of symmetric functions. 

\begin{corollary}\label{corollary:composition}
Let $h: \bin^m \ra \bin$ be any function and $g_i: \bin^k \ra \bin$ ($1 \leq i \leq m$) be symmetric functions. 
Set $n = mk$. 
Define a function $f: \bin^n \ra \bin$ as follows: 
\[
f(x_1, x_2, \ldots, x_n) = h(g_1(x_1, \ldots, x_k), g_2(x_{k+1}, \ldots, x_{2k}), \ldots, g_m(x_{n-k+1}, \ldots, x_n)).
\]
Then, there exists an LQR-PSM protocol for $f$ with communication complexity $n\cdot \log_2 P_L = O(n\cdot L)$ for $L = (k+1)^{n/k}$. 
\end{corollary}

\begin{proof}
We can observe that the function $f$ can be embedded to a linear function $g: \Z^n \ra \Z$ in the following:
\[
g(x_1, x_2, \ldots, x_n) = \sum_{i=1}^{k} x_i + \sum_{i=k+1}^{2k} (k+1)x_i + \sum_{i=2k+1}^{3k} (k+1)^2 x_i + \cdots + \sum_{i=(m-1)k+1}^{mk} (k+1)^{m-1}x_i.
\]
We have 
\[
l(g) = 1 + k + (k+1)k + (k+1)^2k + \cdots + (k+1)^{m-1}k = (k+1)^{n/k}.
\]
From Theorem \ref{theorem:LQR-PSM}, we have an LQR-PSM protocol with communication complexity $n\cdot \log_2 P_L = O(n\cdot L)$ for $L = (k+1)^{n/k}$. 
\end{proof}

\begin{remark}
By setting $(m, k)= (1, n)$, we obtain Corollary \ref{corollary:symmetric}. 
By setting $(m, k)= (n, 1)$, we obtain Corollary \ref{corollary:any}. 
In this sense, Corollary \ref{corollary:composition} is a generalization of Corollaries \ref{corollary:symmetric} and \ref{corollary:any}. 
\end{remark}

By computer experiment, we search LQR-PSM protocols for several symmetric functions with minimum communication complexity in Table \ref{tab:protocol_list}: AND is the logical AND function, XOR is the logical exclusive OR function, EQ is a function that outputs $1$ if and only if all bits are equal, and MAJ is a function which outputs $1$ if and only if half or more bits are $1$. 
Note that these protocols are more efficient than those of Corollary \ref{corollary:symmetric}. 

\begin{table}[htp]
\caption{The list of LQR-PSM protocols for AND, XOR, EQ, and MAJ.}
\begin{tabular}{c|cccc}\hline
$n$ & AND & XOR & EQ & MAJ\\ \hline
$2$ & $[2,1,1]_5$ & $[2,2,4]_5$ & $[1,1,2]_5$ & $[2,2,2]_5$ \\
$3$ & $[6,1,1,1]_{11}$ & $[6,3,3,3]_7$ & $[1,1,1,1]_5$ & $[3,3,3,2]_7$\\
$4$ & $[5,1,1,1,1]_{13}$ & $[12,1,1,1,7]_{17}$ & $[5,1,1,1,1]_{11}$ & $[6,2,2,2,2]_{11}$ \\
$5$ & $[11,1,1,1,1,1]_{41}$ & $[14,2,2,2,2,2]_{19}$ & $[4,1,1,1,1,1]_{13}$ & $[6,2,2,2,2,2]_{11}$ \\
$6$ & $[18,1,1,1,1,1,1]_{53}$ & $[15,1,1,1,1,1,6]_{41}$ & $[10,1,1,1,1,1,1]_{41}$ & $[21,3,3,3,3,3,3]_{31}$ \\
$7$ & $[52,1,1,1,1,1,1,1]_{83}$ & $[35,1,1,1,1,1,1,1]_{79}$ & $[17,1,1,1,1,1,1,1]_{53}$ & $[21,3,3,3,3,3,3,2]_{31}$\\ \hline
\end{tabular}
\label{tab:protocol_list}
\end{table}

\subsection{QR-PSM Protocols from DREs}\label{ss:construction}

In this subsection, we construct QR-PSM protocols from DREs. 

\begin{theorem}\label{theorem:construction}
Let $f: \bin^n \ra \bin$ be a function, and $g: \Z^n \ra \Z$ an embedding of $f$. 
Let $h: \Z^{n+2} \ra \Z$ be a function such that $h(\x_1, \x_2, \ldots, \x_{n+2}) := (g(\x_1, \x_2, \ldots, \x_n) + \x_{n+1}) \cdot \x_{n+2}$. 
Then, there exists a QR-PSM protocol computing $f$ with communication complexity $O(\dare{h} \cdot l(g))$. 
\end{theorem}

\begin{proof}
From Theorem \ref{thm:peralta} and Corollary \ref{cor:peralta}, there exists a prime $p$ with $\log_2 p = O(l(g))$ containing every $l(g)$-bit string. 
Since $g(x) = g(x')$ implies $f(x) = f(x')$, we can take an offset $a_0 \in \Z_p$ such that $\Qr{p}{a_0 + g(x)} = (-1)^{f(x)}$ for all $x \in \bin^n$. 

From the assumption of the statement, there exists a DRE of $h = (g + x_{n+1}) \cdot x_{n+2}$ with output length $\dare{h}$. 
Set $s := \dare{h}$. Let $\hat{h}: \Z_p^{n+2} \times \Z_p^m \ra \Z_p^s$ be the DRE of $h = (g + x_{n+1}) \cdot x_{n+2}$ with output complexity $s$. 
It has the following form:
\[
\hat{h}((x_1, x_2, \ldots, x_{n+2}), r) = (\hat{h}_0(r), \hat{h}_1(x_1, r), \hat{h}_2(x_2, r), \ldots, \hat{h}_{n+2}(x_{n+2}, r))
\]
where $\hat{h}_0: \Z_{p}^m \ra \Z_{p}^{s_0}$ and $\hat{h}_i:  \Z_{p} \times  \Z_{p}^m \ra \Z_{p}^{s_i}$ ($1\leq i \leq n+2$) are functions such that $\sum_{i=0}^{n+2} s_i = s$. 
Let ${\sf dec}: \Z_{p}^s \ra \Z_{p}$ be the decryption function of the DRE. 

The QR-PSM protocol $\Pi = (n, \bin^n, \{-1, 0, 1\}, R, M, (\Enc_i)_{1\leq i \leq n}, \Dec)$ modulo $p$ is defined as follows:
\begin{itemize}
\item $M_1 = \Z_{p}^{s_0 + s_1 + s_{n+1} + s_{n+2}}$ and $M_i =  \Z_{p}^{s_i}$ for all $2 \leq i \leq n$. 
\item $R = \Z_{p}^m \times \QR_{p}$. (Recall that $\QR_{p}$ is the set of nonzero quadratic residues modulo $p$). 
\item $\Enc_1(x_1, (r, r')) = (\hat{h}_0(r), \hat{h}_1(r' \cdot x_1, r), \hat{h}_{n+1}(a_0, r), \hat{h}_{n+2}(r', r))$ and $\Enc_i(x_i, (r, r')) = \hat{h}_i(x_i, r)$ for $2 \leq i \leq n$, where $r \in \Z_{p}^m$ and $r' \in \QR_{p}$. 
\item $\Dec(m_1, m_2, \ldots, m_n) = \Qr{p}{{\sf dec}(m_1, m_2, \ldots, m_n)}$. 
\end{itemize}
The correctness of the protocol follows from the correctness of the DRE, i.e., ${\sf dec}(m_1, m_2, \ldots, m_n) = (g(x) + a_0) \cdot r'$. 
The security of the protocol follows from the security of the DRE $\hat{h}$ directly. 
The communication complexity of the protocol is $s\log_2 p = O(\dare{h} \cdot l(g))$. 
\end{proof}

\begin{corollary}\label{corollary:formula}
Let $f: \bin^n \ra \bin$ be a function which is embedded into a degree-$d$ polynomial $g: \Z^n \ra \Z$ having $m$ terms. 
Then, there exists a QR-PSM protocol computing $f$ with communication complexity $m^2 \cdot d \cdot 2^{O(\sqrt{\log d})}$. 
\end{corollary}

\begin{proof}
Let $h: \Z^{n+2} \ra \Z$ be a function defined by $h := (g + x_{n+1}) \cdot x_{n+2}$. 
By expanding the formula, $h$ can be regarded as a degree-$(d+1)$ polynomial having $m+1$ terms. 
From Theorem \ref{thm:DRE_poly}, we have a DRE of $h$ with output length $m\cdot d \cdot 2^{O(\sqrt{\log d})}$. 
From Theorem \ref{theorem:construction}, we have a QR-PSM protocol computing $f$ with communication complexity $O(\dare{h} \cdot l(g)) = m^2 \cdot d \cdot 2^{O(\sqrt{\log d})}$ since the embedding length of $g$ is $l(g) = m + 1$. 
\end{proof}

\section{Upper Bound on Primes for QR-PSM Protocols}\label{s:prime}

\subsection{LQR-PSM Primes}\label{ss:prime_def}

We define \emph{the $n$-th linear QR-PSM (LQR-PSM) prime $L_n$} as the smallest prime $p$ such that for any function $f: \bin^n \ra \bin$, there exists a linear QR-PSM protocol modulo $p$ computing $f$. 
The $n$-th LQR-PSM prime for $1 \leq n \leq 4$ are: $L_1 = 3, L_2 = 7, L_3 = 11$, and $L_4 = 37$. 
Although $P_i = L_i$ for $1 \leq n \leq 4$, it does not hold in general. 
Indeed, from Theorem \ref{theo:lowerbound} and Corollary \ref{cor:peralta}, we have $L_n > P_n$ for sufficiently large $n$. 

An LQR-PSM prime is upper bounded by a Peralta prime. 
A trivial bound is $L_n \leq P_{2^n}$ since the length of the truth table is $2^n$. 
The following lemma gives a somewhat non-trivial bound on LQR-PSM primes. 

\begin{lemma}\label{lem:universal_peralta}
We have $L_n \leq P_{2^{n-1}}$. 
\end{lemma}

\begin{proof}
Set $p = P_{2^{n-1}}$. 
Let $f: \bin^n \ra \bin$ be any function. 
For a bit $b \in \bin$, let $f_b: \bin^{n-1} \ra \bin$ be a function such that $f_b(x_1, x_2, \ldots, x_{n-1}) = f(x_1, x_2, \ldots, x_{n-1}, b)$, and $t_b \in \bin^{2^{n-1}}$ a string such that the $i$-th bit ($0 \leq i < 2^{n-1}$) of $t_b$ is $f_b(i_1, i_2, \ldots, i_{n-1})$ if $i = \sum_{j=1}^{n-1} 2^{j-1}i_j$, i.e., $t_b$ is the truth table of $f_b$. 
From the property of Peralta prime, $\seq_p$ contains both $t_0$ and $t_1$. 
Let $b_0, b_1 \in \Z_p$ be the offset of the truth tables $t_0, t_1$, i.e., $t_0$ (resp. $t_1$) starts at the $b_0$-th (resp. the $b_1$-th) bit of $\seq_p$. 
Without loss of generality, we can assume $b_0 \leq b_1$. 
Now we have a LQR-PSM protocol $[a_0, a_1, a_2, \ldots, a_n]_p$ computing $f$, where $a_0 = b_0$, $a_i = 2^{n-1-i}$ for $1 \leq i \leq n-1$, and $a_n = b_1 - b_0$. 
Therefore, we have $L_n \leq P_{2^{n-1}}$. 
\end{proof}

We obtain a lower bound on LQR-PSM primes via counting the number of LQR-PSM protocols. 

\begin{theorem}\label{theo:lowerbound}
We have $L_n \geq 2^{\frac{2^n-2}{n}}$. 
\end{theorem}

\begin{proof}
We say that two protocols $[a_0, a_1, \ldots, a_n]_p$ and $[b_0, b_1, \ldots, b_n]_p$ are \emph{conjugate} if there exists a quadratic residue $s \in \QR_p$ such that $b_i = sa_i$ for $0 \leq i \leq n$. 
Note that if two protocols are conjugate, they compute the same function. 
Since the number of $n$-variable Boolean functions $2^{2^n}$ is a lower bound on the number of protocols $\frac{2p^{n+1}}{p-1}$ (up to conjugate), we have $\frac{2p^{n+1}}{p-1}\geq 2^{2^n}$. 
Since it holds $4 \geq \frac{2p}{p-1}$ for every prime $p$, we have $4p^{n}\geq 2^{2^n}$. 
Taking logarithms, we have $p \geq 2^{\frac{2^n-2}{n}}$. 
\end{proof}

\subsection{Paley Graphs and Paley Tournaments}\label{ss:Paley}

We introduce Paley graphs and Paley tournaments, which play important roles in many areas, such as graph theory and additive combinatorics.
In this paper, a {\it graph} is an undirected graph without multiple edges and loops, and a {\it tournament} is an oriented complete graph.  

\begin{definition}[Paley graph]
Let $p\equiv 1 \pmod{4}$ be a prime. Then, the {\it Paley graph} $G_p$ with $p$ vertices is a graph with vertex set $\Z_p$ in which two distinct vertices $x$ and $y$ are adjacent if and only if $x-y \in \mathcal{R}_p$.
\end{definition}
Note that the adjacency of $x,y$ is independent of the order of $x, y$ since $\qr{-1}=1$, which follows from the assumption of $p$. 

\begin{definition}[Paley tournament]
Let $p\equiv 3 \pmod{4}$ be a prime. Then, the {\it Paley tournament} $T_p$ with $p$ vertices is a tournament with vertex set $\Z_p$ in which for two distinct vertices $x$ and $y$, there is a directed edge from $x$ to $y$ if and only if $x-y \in \mathcal{R}_p$.
\end{definition}

Note that the Paley tournament is a tournament, i.e., every distinct vertex $x, y$ have either a directed edge from $x$ to $y$ or a directed edge from $y$ to $x$ since $\qr{-1}=-1$ holds, which follows from the assumption of $p$.

Figure \ref{fig:one} shows Paley graph $G_{17}$ and Figure \ref{fig:two} shows Paley tournament $T_7$ as examples. 

\begin{figure}[htbp]
 \begin{minipage}{0.5\hsize}
  \begin{center}
   \includegraphics[width=50mm]{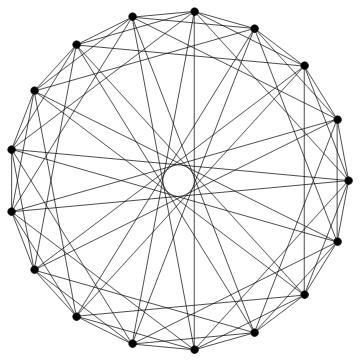}
  \end{center}
  \caption{$G_{17}$}
  \label{fig:one}
 \end{minipage}
 \begin{minipage}{0.5\hsize}
  \begin{center}
   \includegraphics[width=50mm]{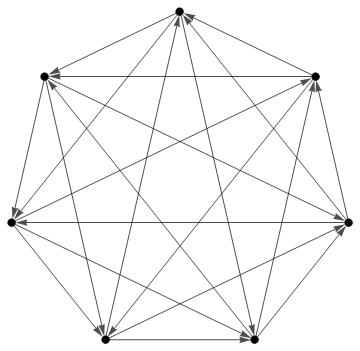}
  \end{center}
  \caption{$T_7$}
  \label{fig:two}
 \end{minipage}
\end{figure}

Paley graph (tournament) is known as a typical example of graphs (tournaments) satisfying various ``random-like" properties, which means properties that random graphs (tournaments) realize with high probability~\cite[Chapter 9]{AS2016}, \cite{B2009, S2021}.

The following property is one of such random-like properties.

\begin{definition}
\label{def-ec}
Let $n\geq 1$ be an integer.
Then, a graph $G$ with vertex set $\Z_p$ is said to have the property $(*)_n$ if for any set $S$ of $n$ consecutive elements of $\Z_p$ and any pair of disjoint (possibly empty) sets of elements, say $A$ and $B$, with $A\cup B=S$, there exists a vertex $z_{A, B} \notin S$ such that $z_{A, B}$ is adjacent to all vertices in $A$, but none in $B$. 
Similarly, a tournament $T$ with vertex set $\Z_p$ is said to have the property $(*)_n$
if for any set $S$ of $n$ consecutive elements of $\Z_p$ and any pair of disjoint (possibly empty) sets of elements $A$ and $B$ with $A\cup B=S$, there exists a vertex $z_{A, B} \notin S$ such that for every vertex $a \in A$ and $b \in B$, there exist an edge from $z_{A, B}$ to $a$ and an edge from $b$ to $z_{A, B}$. 
\end{definition}

\begin{remark}
The property $(*)_n$ is a weaker version of the $n$-existentially closed ($n$-e.c.) property which is known as a finite-analogue of the axiom of the countable random graph (a.k.a. the Rado graph, see, e.g., \cite{C1997}).
The details of the $n$-e.c. property and its application to constructing circulant almost orthogonal arrays can be found in \cite{B2009} and \cite{YSPS2022}.
\end{remark}

\subsection{Upper Bound on Peralta Primes}\label{ss:upper}

The following theorem establishes a connection between Paley graphs, tournaments and Peralta primes.  
The fundamental idea to prove this theorem can be found in \cite{YSPS2022}.
\begin{theorem}
\label{thm-main1}
Let $n\geq 1$ be an integer and $p>n$ denote an odd prime.
When $p\equiv 1 \pmod{4}$, $p$ is $n$-Peralta if and only if $G_p$ has the property $(*)_n$.
Similarly, when $p\equiv 3 \pmod{4}$, $p$ is $n$-Peralta if and only if $T_p$ has the property $(*)_n$.
\end{theorem}

\begin{proof}
Let $n\geq 1$ and assume that $p\equiv 1 \pmod{4}$ is a prime with $p>n$; the discussion below works for the case of a prime $p\equiv 3 \pmod{4}$ as well.
Suppose that the Paley graph $G_p$ has the property $(*)_n$.
Let $t:=(t_1, t_2, \ldots, t_n)\in \bin^n$ be an arbitrarily given sequence.
Set $A:=\{i \in \{1, 2, \ldots, n\} \mid t_i=1\}$ and $B:=\{i \in \{1, 2, \ldots, n\} \mid t_i=0\}$.
Notice that $A\cup B=\{1, 2, \ldots, n\}$.
Then, from the assumption of $G_p$, there exists some $z=z_{A, B} \in \Z_p \setminus \{1, 2, \ldots, n\}$ such that $(\frac{i-z}{p})=1$ if and only if $i \in A$. Here notice that for any $i \in \{1, 2, \ldots, n\}$ we have $(\frac{i-z}{p})\neq 0$ since $z \notin \{1, 2, \ldots, n\}$. 
Now consider the sequence 
\[
\seq_{z}:=\biggl( \frac{1}{2}+\frac{1}{2}\Qr{p}{1-z}, \frac{1}{2}+\frac{1}{2}\Qr{p}{2-z}, \ldots, \frac{1}{2}+\frac{1}{2}\Qr{p}{n-z} \biggr).
\]
Since $z \notin \{1, 2, \ldots, n\}$, $\seq_{z}$ forms a consecutive subsequence of $\seq_p$, and we now have $\seq_{z}=t$.
Conversely if $p$ is an $n$-peralta prime, then $\seq_{p}$ contains any sequence $t \in \{0, 1\}^n$. Since the permutation $x \mapsto x+1$ on $\Z_p$ is an automorphism of $G_p$, to prove that $G_p$ has the property $(*)_n$, it suffices to check that there exists $z_{A, B}$ with respect to the subsets $A, B$ defined above, which is obvious from the assumption of $p$.
\end{proof}

Thus, we immediately obtain the following corollary.

\begin{corollary}
\label{cor-main1}
For $n\geq 1$, let $m^{(G)}_n$ be the least prime $p \equiv 1 \pmod{4}$ such that $G_p$ has the property $(*)_n$, and similarly, $m^{(T)}_n$ denotes the least prime $p \equiv 3 \pmod{4}$ such that $T_p$ has the property $(*)_n$.
Set $m_n:=\min\{m^{(G)}_n, m^{(T)}_n\}$.
Then, we have $P_n=m_n$.
\end{corollary}

Substantially, the following theorem was proved by Graham and Spencer~\cite{GS1971}, Blass, Exoo and Harary~\cite{BEH1981}, Bollob\'{a}s and Thomason~\cite{BT1981} in the context of graph theory.

\begin{theorem}[\cite{BEH1981, BT1981, GS1971}]
\label{thm-ec}
For $n\geq 1$ and every prime $p>n^22^{2n-2}$, both of $G_p$ and $T_p$ have the property $(*)_n$.
In particular, $m_n> n^22^{2n-2}$ for $n\geq 1$.
\end{theorem}

Furthermore, it was proved in \cite{AC1993g, AC1993t} that for an odd prime $p$, both of $G_p$ and $T_p$
have the property $(*)_n$ if $p>\{(n-3)2^{n-1}+2\}\sqrt{p}+(n+1)2^{n-1}-1$.

The following corollary is a direct consequence of Theorems~\ref{thm-main1}, \ref{thm-ec} and Corollary~\ref{cor-main1}, which improves Corollary \ref{cor:peralta} by a constant factor $(1+\sqrt{2})^2 \fallingdotseq 5.828$. 

\begin{corollary}
\label{cor-main2}
If an odd prime $p$ satisfies that $p> n^22^{2n-2}$, then $p$ is $n$-Peralta.
As a consequence, we have $P_n< n^22^{2n-2}$ for $n\geq 1$.
\end{corollary}

Applying the Baker-Harman-Pintz theorem on prime gaps in \cite{BHP2001}, we obtain the following corollary. 

\begin{corollary}
\label{cor-primegap}
For any sufficiently large $n$, there exists an $n$-Peralta prime $p$ with $p \in [n^22^{2n-2}, n^22^{2n-2}+(n^22^{2n-2})^{0.525}]$, which means that $p=(1+o(1))n^22^{2n-2}$.
\end{corollary}

\begin{remark}
A modification of the proof of \cite[Theorem 4.1]{cost2010} shows that for each $n\geq 1$ there is a graph (and tournament) with vertex set $\Z_p$ such that $p=O(n2^n)$ satisfying the property $(*)_n$, where such a graph can be constructed from random Cayley graphs over $\Z_p$.
Since it is known that the Paley graph $G_p$ has various properties that random Cayley graphs over $\Z_p$ satisfy with high probability,
we guess that in fact $m_n=P_n=o(n^2 2^{2n})$.
Although at present there seems to be no known direct approach toward this conjecture, it may be possible to obtain some supporting evidences by considering the following ``random" graph, for example. 
Suppose that $p\equiv 1 \pmod{4}$ is a prime and $1/2\leq q\leq 1$ is a real number.
Then the set $\mathcal{R}_p$ can be partitioned into two non-empty sets $\mathcal{R}_p^{+}$ and $\mathcal{R}_p^{-}$ with same size such that $\mathcal{R}_p^{-}=\{-r \mid r\in \mathcal{R}_p^{+}\}$.
Then for each $r \in \mathcal{R}_p^{+}$ choose a pair $\{r, -r\}$ independently with probability $q$ and form the set $U_p \subseteq \mathcal{R}_p$ consisting of all chosen quadratic residues in $\mathcal{R}_p^+$ and their additive inverses in $\mathcal{R}_p^{-}$. Then construct a graph (denoted by $G_p(q)$) with vertex set $\Z_p$ and connect two vertices $x$ and $y$ if and only if $x-y \in U_p$. 
(A similar construction for primes $p \equiv 3 \pmod{4}$ can be established as well.)
Notice that $G_p(q)$ is a spanning subgraph of $G_p$, and the ``closer'' $G_p(q)$ is to $G_p$, the closer $q$ is to $1$ (in particular $G_p(q)=G_p$ if $q=1$).
We believe that for $q=1-\varepsilon$ with {\it any} $\varepsilon>0$ the probability that $G_p(q)$ with $p=o(n^22^{2n})$ has the property $(*)_n$ tends to $1$ (as $n\to \infty$).
At present it is possible to confirm this claim for $q<3/4$ which follows from a simple computation showing that the probability that $G_p(q)$ does not have the property $(*)_n$ is at most $2^n(1-(1-q)^n)^{p-n}$, which however seems to be far from sharp estimations.
\end{remark}

\section*{Acknowledgments}

The first author was supported during this work by JSPS KAKENHI Grant Number JP21K17702. 
The second author was supported during this work by JSPS KAKENHI Grant Number JP20J20797. 
The third author was supported during this work by JSPS KAKENHI Grant Number JP20J00469. 
The last author was supported during this work by JSPS KAKENHI Grant Number JP19H01109, JST CREST Grant Number JPMJCR2113, and JST AIP Acceleration Research JPMJCR22U5. 

\bibliographystyle{abbrv}
\bibliography{psm}

\end{document}